\documentclass[%
superscriptaddress,
preprint,
 amsmath,amssymb,
 aps,
pra,
]{revtex4-2}

\usepackage{graphicx}% Include figure files
\usepackage{dcolumn}% Align table columns on decimal point
\usepackage{bm}% bold math
\usepackage{braket}
\usepackage{comment} % comment
\usepackage{hyperref}% add hypertext capabilities
\usepackage{siunitx}
\usepackage{xcolor}

\usepackage{nicematrix}
\NiceMatrixOptions{cell-space-limits = 1pt}

\usepackage{mathtools}

\usepackage{amsthm}
\newtheorem{theorem}{Theorem}
\newtheorem*{theorem*}{Theorem}

\newtheorem{algorithm}{Algorithm}

\newtheorem{algorithmappendix}{Algorithm}[section]

\theoremstyle{definition}

\theoremstyle{remark}

\newtheorem*{example*}{Example}
\newtheorem*{algorithm*}{Algorithms}

\usepackage{etoolbox}

\makeatletter
\patchcmd{\frontmatter@abstract@produce}
  {\vskip200\p@\@plus1fil
   \penalty-200\relax
   \vskip-200\p@\@plus-1fil}
  {}
  {}
  {}
\makeatother

\begin{document}

%\preprint{Cheng Guo's preprint}

\title{Unitary Control of Multiport Wave Transmission}% Force line breaks with \\

\author{Cheng Guo}
\email{guocheng@stanford.edu}
\affiliation{Ginzton Laboratory and Department of Electrical Engineering, Stanford University, Stanford, California 94305, USA}

\author{David A. B. Miller}
\affiliation{Ginzton Laboratory and Department of Electrical Engineering, Stanford University, Stanford, California 94305, USA}

\author{Shanhui Fan}
\email{shanhui@stanford.edu}
\affiliation{Ginzton Laboratory and Department of Electrical Engineering, Stanford University, Stanford, California 94305, USA}%

\date{\today}% It is always \today, today,

\begin{abstract}Controlling wave transmission is crucial for various applications. In this work, we apply the concept of unitary control to manipulate multiport wave transmission. Unitary control aims to control the behaviors of a set of orthogonal waves simultaneously. The approach fully harnesses the capability of wavefront shaping techniques, with promising applications in communication, imaging, and photonic integrated circuits. Here we present a detailed theory of unitary control of wave transmission, focusing on two key characteristics: total (power) transmittance and direct (field) transmission. The total transmittance for an input port represents the fraction of total transmitted power with respect to  the input power for wave incident from an input port. The direct transmission for an input port denotes the complex transmission amplitude from that input port to its corresponding output port. We address two main questions: (i) the achievable total transmittance and direct transmission for each port, and (ii) the configuration of unitary control to attain desired transmission values for each port. Our theory illustrates that unitary control enables uniform total transmittance and direct transmission across any medium. Furthermore, we show that reciprocity and energy conservation enforce direct symmetry constraints on wave transmission in both forward and backward directions under unitary control.
\end{abstract}
\maketitle

%\tableofcontents

\section{introduction}\label{sec:introduction}

Controlling the transmission of waves through complex media is crucial for various applications in imaging~\cite{sebbah2001a,ntziachristos2010a,cizmar2012,kang2015,guo2018,guo2018a,yoon2020b,wang2020p,long2021,bertolotti2022,wang2022,long2022}, communications~\cite{miller2013c,miller2019,seyedinnavadeh2024}, sensing~\cite{aulbach2011,sarma2015,mounaix2016,jeong2018a,liu2020s}, photonic integrated circuits~\cite{turpin2018,leedumrongwatthanakun2020,wetzstein2020a}, and spatiotemporal wave shaping~\cite{wang2021c,guo2021c,shen2023}. Recent advancements~\cite{guan2012,katz2012,horstmeyer2015,vellekoop2015,yu2015d,pai2021b} have allowed for precise control over coherent wave transport in complex media~\cite{rotter2017,cao2022a} including highly scattering biological tissues~\cite{yaqoob2008,ntziachristos2010a,horstmeyer2015,yu2015d} and multimode optical fibers~\cite{fan2005,cizmar2011,cizmar2012,papadopoulos2012,carpenter2015,xiong2016}, enabling applications such as spatial and temporal focusing~\cite{lerosey2007,vellekoop2008b,katz2011,mccabe2011,xu2011a,papadopoulos2012,park2013a,horstmeyer2015,blochet2017,jeong2018a}, transmittance enhancement and suppression~\cite{vellekoop2008,popoff2010,aulbach2011,kim2012a,shi2012,yu2013a,gerardin2014,pena2014a,popoff2014,davy2015,shi2015a,bender2020c}, and optical micro-manipulation~\cite{cizmar2010a,gonzalez-ballestero2021,hupfl2023}.

A key development in wave transmission control has been the use of wavefront shaping techniques, particularly spatial light modulators (SLMs)~\cite{vellekoop2007,yu2017e}. SLMs can adjust the phase of reflected light to tailor a coherent input field into a customized wavefront, resulting in the desired transmitted pattern after passing through the complex medium. This process, known as coherent control~\cite{popoff2014,liew2016,mounaix2016,guo2024c}, has revolutionized our ability to control wave propagation through complex media.

However, initial works on coherent control through SLMs have been primarily limited to a single incident wave, while many applications require the simultaneous control of multiple orthogonal waves~\cite{berdague1982,luo2014,su2021}. This multiport control is becoming possible with programmable unitary photonic devices like Mach-Zehnder interferometer (MZI) meshes~\cite{reck1994,miller2013c,miller2013a,miller2013b,carolan2015,miller2015,clements2016,ribeiro2016,wilkes2016,annoni2017,miller2017d,perez2017,harris2018,pai2019} and multiplane light conversion devices~\cite{morizur2010,labroille2014,tanomura2022,kupianskyi2023,taguchi2023,zhang2023b}. These devices can perform arbitrary unitary transformations and hold promise for various applications in quantum computing~\cite{carolan2015,carolan2020,elshaari2020,wang2020aw,chi2022,madsen2022,pelucchi2022}, machine learning~\cite{shen2017,prabhu2020,zhang2021f,ashtiani2022,bandyopadhyay2022,ohno2022,chen2023,pai2023}, and optical communications~\cite{clements2016,annoni2017,burgwal2017,melati2017,choutagunta2020,buddhiraju2021}. By transforming from one to another set of orthogonal incident waves, they can achieve sophisticated multiport transmission control. We refer to such control as \emph{unitary control}~\cite{guo2023b,guo2024a,guo2024b}, as it is mathematically represented by a unitary transformation of the input wave space. Refs.~\cite{miller2013c} and~\cite{seyedinnavadeh2024} are explicit physical examples of such unitary control.

In this paper, we present a systematic theory of unitary control over multiport wave transmission. Some optimum results are already known from the singular-value decomposition SVD approach to waves and optics~\cite{miller2000,miller2012}. Such SVD optimum orthogonal channels between inputs and outputs, which can be thought of as ``communication modes"~\cite{miller2000,miller2019} or ``mode-converter basis sets"~\cite{miller2012,miller2019}, can be found and implemented physically using unitary control from MZI meshes~\cite{miller2013c,miller2019,seyedinnavadeh2024}. Beyond these SVD modes, other sets of orthogonal waves are relevant in various contexts. For instance, principal modes in multimode waveguides exhibit minimal modal dispersion and form orthogonal bases at both waveguide ends~\cite{fan2005}. Additionally, orbital angular momentum (OAM) modes are common in multimode fiber communications~\cite{yue2012a,bozinovic2013b,richardson2013a,brunet2014b,lian2022}. Understanding the transmission properties of these diverse orthogonal wave sets is thus valuable.

When analyzing the transmission properties of orthogonal waves, several key characteristics are essential for each input wave basis: total transmittance, direct transmission, and cross talk. These quantities are defined after selecting orthogonal bases at both input and output ends of the medium (e.g., a waveguide). For an incident wave in the $i$-th input port, the total transmittance $T_i$ represents the fraction of power transmitted with respect to the input power. The direct transmission $t_i$ denotes the complex transmission amplitude to the $i$-th output port. Cross talk, calculated as $T_i - |t_i|^2$, describes the fraction of power transmitted to output ports other than the $i$-th output port. Though a full transmission matrix is required to describe the transmission properties of the medium, the total transmission and the cross-talk usefully characterize the power transmission properties of the $i$-th input basis wave through the medium.

Here we address two fundamental questions for the transmission of orthogonal waves in unitary control: (i) What are the possible values of total transmittance and direct transmission for each input port under unitary control in a given medium? (ii) How can we configure unitary control to achieve specific total transmittance or direct transmission for each port? We answer these questions using matrix inequalities. Our results provide insights into the principles and implementations of the unitary control method.

The rest of this paper is organized as follows. In Sec.~\ref{sec:notations}, we summarize useful mathematical notations. In Sec.~\ref{sec:theory}, we develop a general theory of unitary control over wave transmission. In Sec.~\ref{sec:applications}, we discuss the physical applications of our theory. We conclude in Sec.~\ref{sec:conclusion}. Appendices contain additional information. Appendices~\ref{appendix:Chu_Fickus}-\ref{Appendix:algorithm3_demo} provide details and demonstrations of the algorithms. Appendix~\ref{appendix:proof} provides detailed mathematical proof.

%\section{Results}

\section{Notations}\label{sec:notations}
We first summarize the notations related to matrices. We denote by $M_{n}$ the set of $n\times n$ complex matrices and  $U(n)$ the set of $n\times n$ unitary matrices. For $M \in M_{n}$, we define $\bm{d} (M) = \left(d_{1}(M),\ldots, d_{n}({M})\right)^{T}$, $\bm{\lambda} (M) = \left(\lambda_{1}(M), \ldots, \lambda_{n}(M)\right)^{T}$, and $\bm{\sigma} (M) = \left(\sigma_{1}(M),\ldots,\sigma_{n}(M)\right)^{T}$ 
the vectors of diagonal entries, eigenvalues, and singular values of $M$~\cite{guo2023c}. We also define 
\begin{equation}
\bm{\sigma}^2(M) \equiv \left(\sigma^2_{1}(M),\ldots,\sigma^2_{n}(M)\right)^{T}.
\end{equation} 
(We take the convention of choosing the singular values to be real numbers, with any complex phase factors instead included in the corresponding singular functions.) Finally, for $\bm{z} = (z_1, \dots, z_n)^T \in \mathbb{C}^n$, we define $|\bm{z}| = (|z_1|, \dots, |z_n|)^T \in \mathbb{R}^n$. (The superscript $T$ here means the transpose, giving a more compact way of writing a column vector. This is not the same ``$T$" as used later for transmission.)

We also discuss the notations of majorization~\cite{marshall2011}. For $\bm{x} = (x_1, x_2, \ldots,x_n)\in \mathbb{R}^n$, we define 
$\bm{x}^\downarrow = (x^\downarrow_1, x^\downarrow_2, \ldots,x^\downarrow_n)$, where $x_1^\downarrow \ge x_2^\downarrow \ge  \cdots \ge x_n^\downarrow$, hence reordering the components of $x$ in non-increasing order. For $\bm{x}=(x_1, \ldots, x_n)$ and $\bm{y} = (y_1, \ldots, y_n)$ in $\mathbb{R}^n$, if
\begin{align}
&\sum_{i=1}^k x_i^\downarrow \le \sum_{i=1}^k y_i^\downarrow, \quad k=1,2,\ldots,n-1;     \\
&\sum_{i=1}^n x_i = \sum_{i=1}^n y_i, \label{eq:def_majorization_equality}
\end{align}
we say that $\bm{x}$ is \emph{majorized} by $\bm{y}$, written as $\bm{x} \prec \bm{y}$. If Eq.~(\ref{eq:def_majorization_equality}) is replaced by a corresponding inequality:
\begin{equation}
\sum_{i=1}^n x_i \leq \sum_{i=1}^n y_i,     
\end{equation}
we say that $\bm{x}$ is \emph{weakly majorized} by $\bm{y}$, written as $\bm{x} \prec_w \bm{y}$. 

\section{Theory}\label{sec:theory}

\begin{figure}[hbtp]
    \centering
    \includegraphics[width=0.5\textwidth]{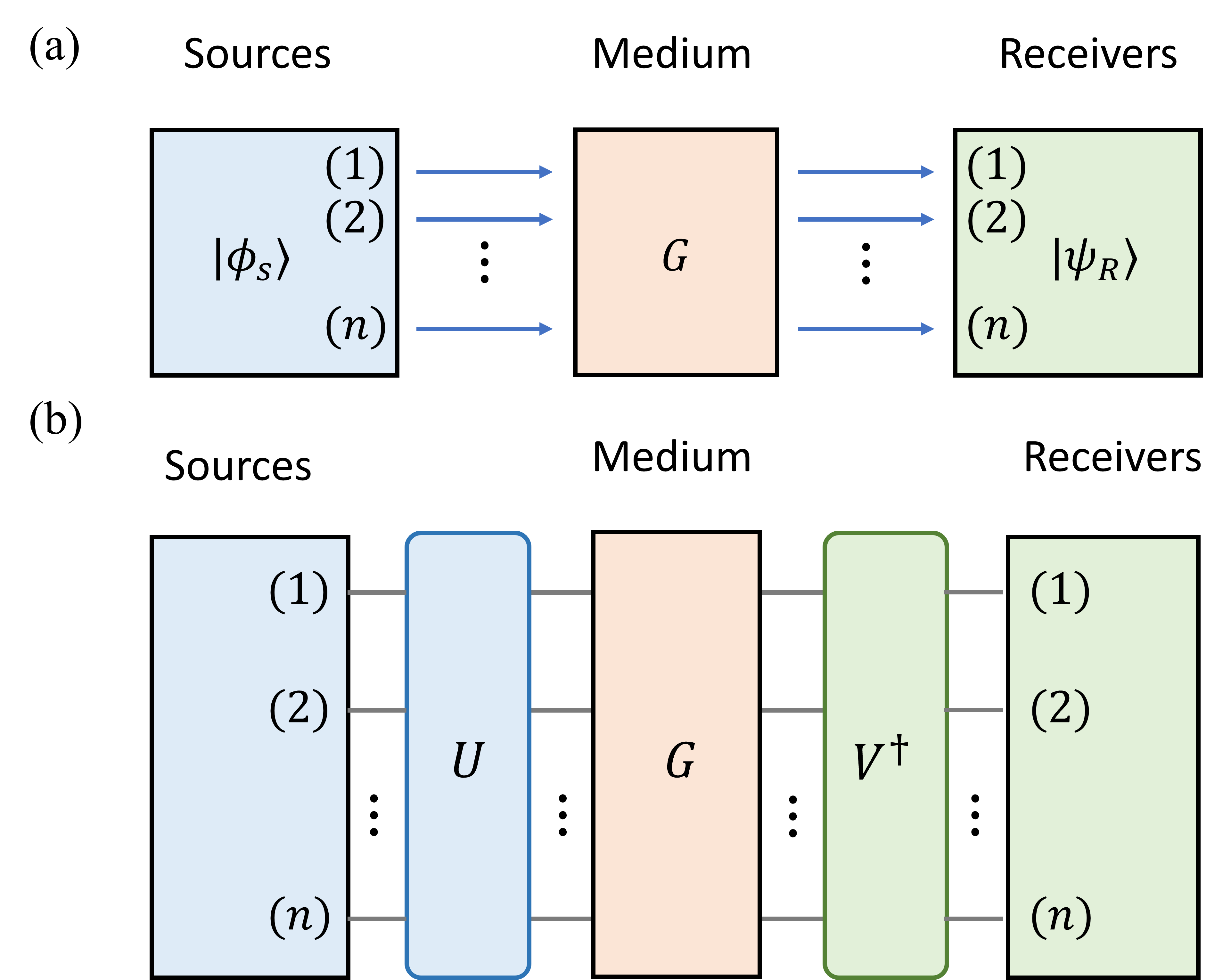}
    \caption{(a) Scheme of wave transmission. An input wave $\ket{\phi_S}$ passes through the medium and becomes the transmitted wave $\ket{\psi_R} = G \ket{\phi_S}$.  (b) Unitary control of wave transmission. Input ports and output ports are numbered from $(1)$ to $(n)$. }
    \label{fig:scheme}
\end{figure}

\subsection{Wave transmission}
Consider a typical setup for wave transmission, where waves are emitted from multiple sources, transmitted through a medium, and detected by several receivers~\cite{miller2019} (see Fig.~\ref{fig:scheme}a). The medium is assumed to be a linear time-invariant system with two sides, each having $n$ ports. We consider a set of input ports, with the wave in the $j$-th input port being written as $\ket{\phi_j^{(i)}}$ and the wave in the $j$-th output port being written as $\ket{\psi_j^{(o)}}$. The wave in one input (or output) port by definition does not overlap with the wave in any other input (or output) port, so this ``port by port" set of input functions is orthogonal, and similarly for the corresponding ``port by port" set of output functions. Using these port functions as bases, we can describe the input and transmitted waves as the vectors of amplitudes
\begin{equation}
\bm{a} = (a_{1},\ldots,a_{n})^{T}, \quad \bm{b} = (b_{1},\ldots,b_{n})^{T},    
\end{equation}
where $a_{i}$ and $b_{i}$ represent the input and transmitted wave amplitudes in their respective ports. The transmission process is described by a complex \emph{transmission matrix} $G\in M_n$ (which can also be considered as a coupling matrix, a device matrix, or a Green's function matrix):~\cite{nazarov1995,vellekoop2008,popoff2010,miller2019}
\begin{equation}
\bm{b} = G\bm{a},    
\end{equation}
where $G_{ij}$ represents the transmission coefficient from the $j$-th port on the left to the $i$-th port on the right. %In general, $G$ is a complex matrix.

\subsection{Unitary control of wave transmission}

Now we introduce unitary control. Unitary control refers to unitarily transforming the input and transmitted wave bases  (Fig.~\ref{fig:scheme}b):
\begin{equation}\label{eq:new_basis}
\ket{\phi_j^{(i)}} \to  \ket{\phi_j^{(i)}[U]} = \sum_{k=1}^{n} U_{kj} \ket{\phi_k^{(i)}}, \quad \ket{\psi_j^{(o)}} \to \ket{\psi_j^{(o)}[V]} =  \sum_{k=1}^{n} V_{kj} \ket{\psi_k^{(o)}},
\end{equation}
where $U, V \in U(n)$. Throughout the paper, we use the square bracket to indicate the dependence on a matrix. The resulting transformed basis functions are themselves now $n$-dimensional vectors of amplitudes. Because unitary operations preserve orthogonality, these new bases are also orthogonal. Under the new bases, the transmission matrix is modified to
\begin{equation}\label{eq:unitary_control_G}
G \to G[U,V] = V^{\dagger} G U,    
\end{equation}
Hence, unitary control corresponds to a unitary equivalence~\cite{horn2012} of a transmission matrix. 

Unitary control also transforms the transmission properties of a system. We focus on two key characteristics: total (power) transmittance and direct (field) transmission. The total transmittance~\cite{popoff2014} is described by a real \emph{total transmittance vector} 
\begin{equation}
\bm{T} \coloneqq (T_{1}, \dots, T_{n})^{T},
\end{equation}
where $T_i\in \mathbb{R}$ represents the total power transmitted when a wave with unit power is incident from the $i$-th input port only. The direct transmission is described by a complex \emph{direct transmission vector}
\begin{equation}
\bm{t}\coloneqq (t_{1}, t_{2}, \dots, t_{n})^{T},   
\end{equation}
where $t_{i}\in \mathbb{C}$ represents the amplitude transmission coefficient from the $i$-th input port to the $i$-th transmitted port. These input and transmitted ports can be chosen freely, but in many practical scenarios, there is a natural choice. For example, in a multimode fiber, the orbital angular momentum (OAM) modes are often chosen as input and transmitted bases~\cite{yue2012a,bozinovic2013b,richardson2013a,brunet2014b,lian2022}. Then the direct transmission describes the transmission coefficients between the same OAM modes. Because the OAM modes are transmitted without mixing in such an ideal fiber, $\bm{T}$ just becomes the vector of power transmissions of these modes, and $\bm{t}$ just becomes the vector of the amplitude transmissions of these modes. $G$ and $G^\dagger G$ then become diagonal matrices in this basis, so that, formally, in our notation
\begin{equation}
    \bm{T} = \bm{d} (G^{\dagger}G), \qquad \bm{t} = \bm{d} (G).  
\end{equation}
Under the unitary control as defined in Eq.~(\ref{eq:unitary_control_G}), the total transmittance vector in the new bases is modified to
\begin{equation}\label{eq:def_T_U}
\bm{T} \to \bm{T}[U] \coloneqq \bm{d} (G^{\dagger}[U,V] G[U,V]) =  \bm{d} (U^{\dagger}G^{\dagger}GU),    
\end{equation}
which depends only on $U$. Similarly, the direct transmission vector in the new bases is modified to
\begin{equation}\label{eq:def_t_U_V}
\bm{t} \to \bm{t}[U,V] \coloneqq  \bm{d} (G[U,V]) = \bm{d} (V^{\dagger}GU),
\end{equation}
which depends on both $U$ and $V$.

\subsection{Major questions}

We ask four basic questions: Given a transmission medium under unitary control, (1) What total transmittance vectors are attainable? (2) How to obtain a given total transmittance vector? (3) What direct transmission vectors are attainable? (4) How to obtain a given direct transmission vector? Questions 1 and 3 ask about the capability and limitation of unitary control. Questions 2 and 4 ask for implementation.

We now reformulate these key questions mathematically. Let $G \in M_{n}$ be a given transmission matrix. Question 1: What is the set
\begin{equation}\label{eq:Question1}
\{\bm{T}\} \coloneqq \{\bm{T}[U] \mid U\in U(n)\}?    \end{equation}
Question 2: Given $\bm{T}_{0} \in \{\bm{T}\}$, what is the set
\begin{equation}\label{eq:Question2}
\{U[\bm{T_{0}}]\} \coloneqq \{U\in U(n) \mid  \bm{T}[U] = \bm{T}_{0}\}? 
\end{equation}
Question 3: What is the set
\begin{equation}\label{eq:Question3}
\{\bm{t}\} \coloneqq \{\bm{t}[U,V] \mid U,V \in U(n)\} ?    
\end{equation}
Question 4: Given $\bm{t}_{0}\in \{\bm{t}\}$, what is the set
\begin{equation}\label{eq:Question4}
\{(U,V)[\bm{t}_{0}]\} \coloneqq \{U,V \in U(n) \mid \bm{t}[U,V] = \bm{t}_{0}\}?    
\end{equation}

\subsection{Main results}

Here we provide complete answers to Questions 1-3 and a partial answer to Question 4.

\subsubsection{Answer to Question 1}

\begin{figure}[htbp]
    \centering
    \includegraphics[width=0.45\textwidth]{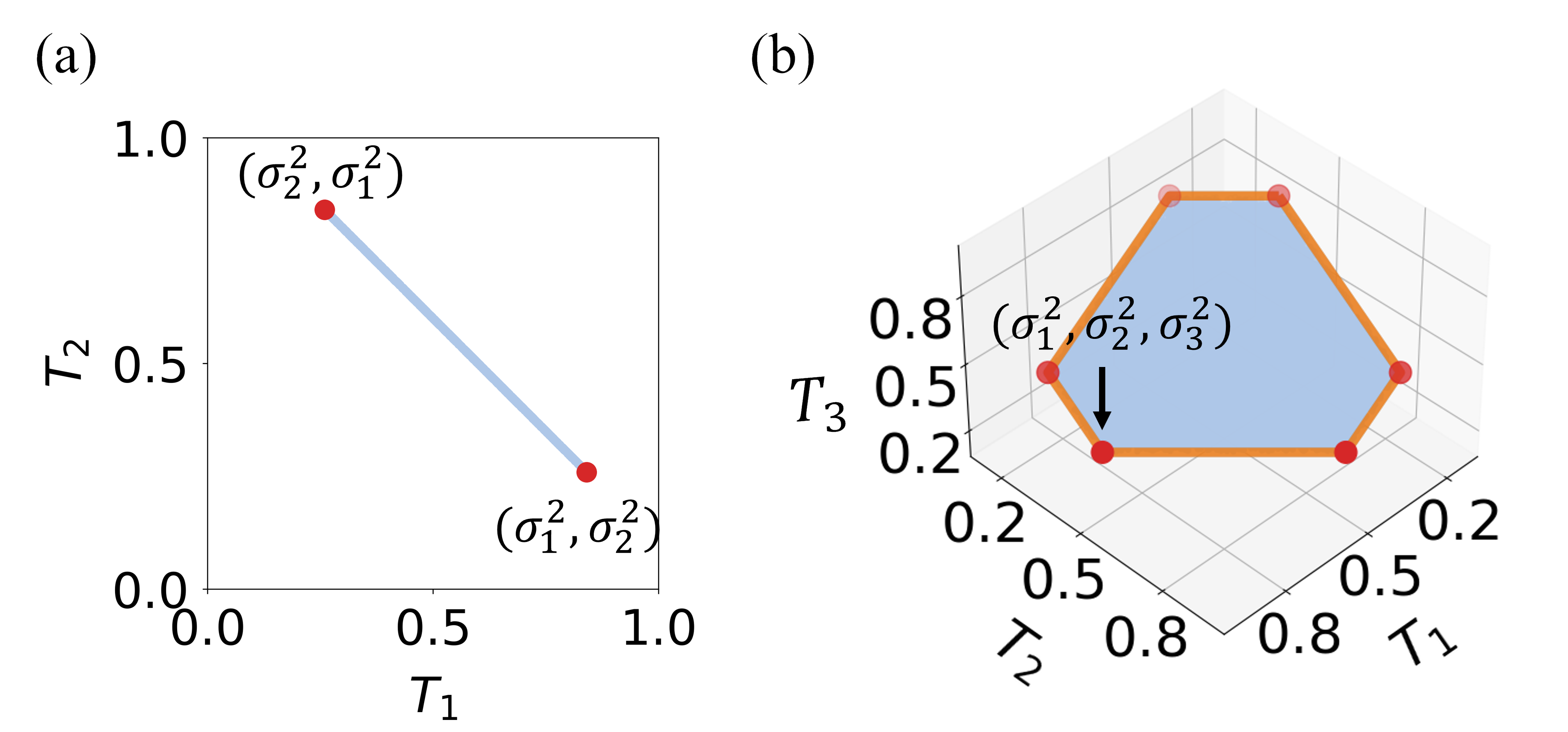}
    \caption{Examples of $\{\bm{T}\}$. (a) $G_2\in M_2$ with $\bm{\sigma}^2(G_2) = (0.84,0.26)^T$. (b) $G_3\in M_3$ with $\bm{\sigma}^2(G_3) = (0.81,0.36,0.16)^T$.} 
    \label{fig:T}
\end{figure}
We start with Question 1. To hint at the solution, we perform two numerical experiments. In the first experiment, we consider a random $2\times 2$ transmission matrix:
\begin{align}\label{eq:G_matrix_2_2_example}
&G_2 = \begin{pmatrix}
0.4-0.5i & 0.1 + 0.3 i \\
0.4 - 0.3 i & 0.5 -0.3i
\end{pmatrix}; \\
&\bm{\sigma}(G_2) = \begin{pmatrix}
0.92 \\ 
0.51
\end{pmatrix}, \quad \bm{\sigma}^2(G_2) = \begin{pmatrix}
0.84 \\
0.26
\end{pmatrix}.
\end{align}
%Note $S\neq S^{T}$, the system is nonreciprocal. 
We generate $\num{1000}$ random $U_{i}\in U(2)$ and calculate $\bm{T}[U_{i}]$ by Eq.~(\ref{eq:def_T_U}). Fig.~\ref{fig:T}(a) shows the result. We see that $\{\bm{T}\}$ is a line segment with endpoints obtained by permuting the coordinates of $\bm{\sigma}^{2}(G_2)$. In the second experiment, we consider a random $3\times 3$ scattering matrix:
\begin{align}\label{eq:G_matrix_3_3_example}
&G_3 = \begin{pmatrix}
-0.14-0.07i & -0.19 - 0.27 i & 0.55 - 0.04i \\
-0.48 - 0.26 i & -0.09 - 0.12i & -0.23 - 0.38i \\
-0.02 - 0.03i & 0.22-0.44i & -0.14 - 0.34 i
\end{pmatrix};  \\ 
&\bm{\sigma}(G_3) = \begin{pmatrix}
0.90 \\ 
0.60 \\
0.40
\end{pmatrix}, 
\quad 
\bm{\sigma}^2(G_3) = \begin{pmatrix}
0.81 \\
0.36 \\
0.16
\end{pmatrix}.
\end{align}
We generate $\num{100000}$ random $U_{i} \in U(3)$ and calculate $\bm{T}[U_{i}]$ by Eq.~(\ref{eq:def_T_U}). Fig.~\ref{fig:T}(b) shows the result. We see that $\{\bm{T}\}$ is a convex hexagon with vertices obtained by permuting the coordinates of $\bm{\sigma}^{2}(G_3)$.

The numerical results above suggest the following observation on the geometry of $\{\bm{T}\}$: For an $n\times n$ transmission matrix $G$, $\{\bm{T}\}$ is a convex subset of an $(n-1)$-dimensional hyperplane in $\mathbb{R}^n$. It is the convex hull spanned by the $n!$ points obtained by permuting the coordinates of $\bm{\sigma}^2(G)$. (The convex hull of a set is the smallest convex set that contains it.) We show that this observation is true as a result of our first theorem: 
 
\begin{theorem}\label{Theorem:T}
Given a transmission matrix $G \in M_n$,  
\begin{align}
\{\bm{T}\} = \{\bm{u} \in \mathbb{R}^n | \bm{u} \prec \bm{\sigma}^{2}(G) \}.   \label{eq:main_result_T} 
\end{align}
\end{theorem}
\begin{proof}
    First, we show  $\bm{T}[U] \in \{\bm{T}\}$ $\implies \bm{T}[U] \prec \bm{\sigma}^{2}(G)$. We use Schur's theorem~\cite{schur1923}: 
\begin{equation}
    \bm{T}[U] = \bm{d}(U^{\dagger}G^{\dagger}GU) \prec \bm{\lambda}(U^{\dagger}G^{\dagger}GU) = \bm{\lambda}(G^{\dagger}G) = \bm{\sigma}^2(G). 
\end{equation}
Second, we show $\bm{u} \prec \bm{\sigma}^{2}(G)$ $\implies \bm{u} \in \{\bm{T}\}$, i.e.,~there exists $U \in U(n)$ such that $\bm{T}[U] = \bm{u}$. We use Horn's theorem~\cite{horn1954}: As $\bm{u}\prec \bm{\sigma}^{2}(G)$, there exists a Hermitian matrix $H$ with $\bm{d}(H) = \bm{u}$ and $\bm{\lambda}(H) = \bm{\sigma}^{2}(G)$. Since $\lambda(G^{\dagger}G) = \bm{\sigma}^{2}(G) = \lambda(H)$, $H$ and $G^{\dagger}G$ are unitarily similar.  Hence there exists $U \in U(n)$ such that $H = U^\dagger G^{\dagger}G U.$ Now we check 
\begin{equation}
    \bm{T}[U] \equiv \bm{d}(U^\dagger G^{\dagger}G U) = \bm{d}(H) = \bm{u}.
\end{equation}   
This completes the proof.
\end{proof}
The geometric observation above is a direct consequence of Theorem~\ref{Theorem:T}. We use Rado's theorem~\cite{rado1952}, which states that for a given $\bm{y} \in \mathbb{R}^n$, the set $\{\bm{x}\in\mathbb{R}^n | \bm{x} \prec \bm{y}\}$ is the convex hull of points obtained by permuting the components of $\bm{y}$. The geometric observation is an application of Rado’s theorem to Eq.~(\ref{eq:main_result_T}). 

It should be noted that the necessary condition, $\bm{T}[U] \in {\bm{T}} \implies \bm{T}[U] \prec \bm{\sigma}^{2}(G)$, is known in literature in some equivalent form~\cite{miller2000,miller2019} but not proven using majorization theory. This implies that no channels can outperform SVD channels~\cite{miller2019}. The sufficient condition, stating that for any $\bm{u} \prec \bm{\sigma}^{2}(G)$, there exists $U \in U(n)$ such that $\bm{T}[U] = \bm{u}$, is a new result. It demonstrates that any power distribution channels not superior to SVD channels in the majorization sense can be achieved through unitary control.

Eq.~(\ref{eq:main_result_T}) is our first main result. It shows that $\{\bm{T}\}$ is completely determined by $\bm{\sigma}(G)$, which is invariant under unitary control: $\bm{\sigma}(V^\dagger G U) = \bm{\sigma}(G)$. We can classify all transmission media by their $\bm{\sigma}$. Two media exhibit the same $\{\bm{T}\}$ if and only if they belong to the same $\bm{\sigma}$ class.

\subsubsection{Answer to Question 2}

Then, we turn to Question 2. The problem corresponds to the following physical scenario.  Suppose we have a medium characterized by a transmission matrix $G$. Given a total transmittance vector $\bm{T}_0\prec \bm{\sigma}^2(G)$, how do we construct the set of all possible unitary control schemes as described by unitary matrices $\{U[\bm{T}_0]\}$ that achieve $\bm{T}_0$? Alternatively, a simpler question, how to construct one unitary control scheme as described by a unitary matrix $U[\bm{T}_0]$ that achieves $\bm{T}_0$?

These two problems can be solved by the following algorithms. We first perform a preparatory step that is common in both algorithms: Suppose $G$ has $p$ distinct singular values, then $G^\dagger G$ has $p$ distinct eigenvalues $\lambda_{1}, \ldots, \lambda_{p}$, with respective multiplicities $n_{1}, \ldots, n_{p}$. Let $\Lambda = \lambda_{1} I_{n_{1}} \oplus \ldots \oplus \lambda_{p} I_{n_{p}}$. We find a $V\in U(n)$ such that $G^\dagger G = V \Lambda V^\dagger$. Now we provide the two algorithms:

\begin{algorithm}[Constructing $\{ U{[\bm{T}_{0}]}\}$]
\hfill
\begin{enumerate}
\item Use Fickus' algorithm~\cite{fickus2013} (see Appendix~\ref{appendix:Chu_Fickus} for details) to construct all Hermitian matrices $H_{i}$ with eigenvalues $\bm{\lambda}(G^\dagger G)$ and diagonal entries $\bm{T}_0$. For each $H_{i}$, find a $V_{i}\in U(n)$ such that $H_{i} = V_{i} \Lambda V_{i}^\dagger$.
    \item We claim that  $U_{i}\in U(n)$ such that $H_{i} = U_{i}^{\dagger} G^\dagger G U_{i}$ if and only if 
\begin{equation}\label{eq:claim_algorithm}
U_{i} = V (W_{1}\oplus\ldots W_{p}) V_{i}^{\dagger}    
\end{equation}
where $W_{k}\in U(n_{k})$, $k=1,\ldots,p$, are arbitrary. Denote the set of all such $U_{i}$ as $\{U_{i}\}$. \label{item:claim1}

\item We claim that $\{U[\bm{T}_{0}]\} = \bigcup_{i} \{U_{i}\}$.\label{item:claim2} (See Ref.~\cite{guo2023b}, SM for proof of the two claims.)
\end{enumerate}\label{alg:all_U} 
\end{algorithm}
\begin{algorithm}[Constructing a ${U[\bm{T}_{0}]}$]
\hfill
\begin{enumerate}
\item Use Chu's first algorithm~\cite{chu1995} (see Appendix~\ref{appendix:Chu_Fickus} for details) to construct a Hermitian matrix $H$ with eigenvalues $\bm{\lambda}(G^\dagger G)$ and diagonal elements $\bm{T}_0$. Find a $V'\in U(n)$ such that $H = V' \Lambda V^{'\dagger}$.
\item We obtain a $U[\bm{T}_0] = V V^{'\dagger}$.
\end{enumerate}\label{alg:a_U}
\end{algorithm}
Algorithms~\ref{alg:all_U} and~\ref{alg:a_U} are our second main result. We illustrate their usage with a numerical example in Appendix~\ref{Appendix:algorithm2_demo}. We consider a $5\times 5$ transmission matrix $G$. Our task is to construct a $U[\bm{T}_0]$ with a randomly assigned goal $\bm{T}_0$. We use Algorithm~\ref{alg:a_U} and complete the task. Importantly, our algorithms allow us to achieve the prescribed total transmittance in \emph{all} ports with a \emph{single} unitary matrix that performs the unitary control. 

\subsubsection{Answer to Question 3}

\begin{figure}[htbp]
    \centering
    \includegraphics[width=0.45\textwidth]{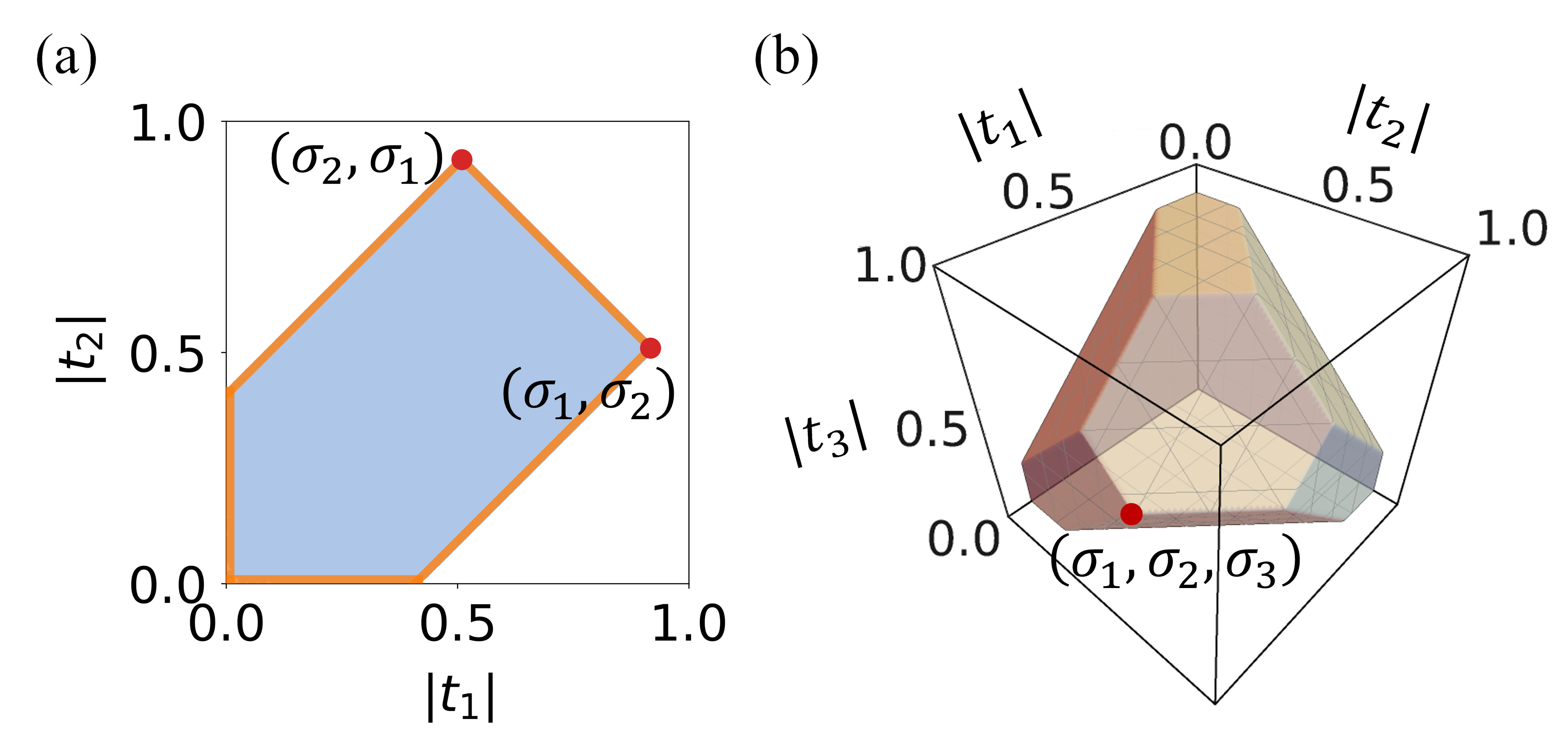}
    \caption{Examples of $\{|\bm{t}|\}$. (a) $G_2\in M_2$ with $\bm{\sigma}(G_2) = (0.92,0.51)^T$. 
    (b) $G_3\in M_3$ with $\bm{\sigma}(G_3) = (0.90,0.60,0.40)^T$.} 
    \label{fig:t_abs}
\end{figure}
Next, we consider Question 3. To hint at the solution, we perform two numerical experiments. In the first experiment, we consider the $2\times 2$ transmission matrix $G_2$ in Eq.~(\ref{eq:G_matrix_2_2_example}). We generate $\num{1000000}$ random $(U_{i},V_{i}) \in U(2)\times U(2)$ and calculate $\bm{t}[U_{i}, V_i]$ by Eq.~(\ref{eq:def_t_U_V}). Fig.~\ref{fig:t_abs}(a) shows the result of $|\bm{t}|[U_{i}, V_i]$~\footnote{It is easy to see that the phases of $\bm{t}$ components can be arbitrary when one considers all possible $(U_i, V_i)$.}. We see that $\{|\bm{t}|\}$ is a convex pentagon, which is defined by the set of solutions to a system of linear inequalities~\footnote{Inequality (\ref{ineq:redundant_2_2}) is redundant but included for later generalization.}:
\begin{align} \label{ineq:redundant_2_2}
|t|^\downarrow_1 &\leq \sigma^\downarrow_1(G_2), \\ |t|^\downarrow_1 + |t|^\downarrow_2 &\leq \sigma^\downarrow_1(G_2) + \sigma^\downarrow_2(G_2),   \\
|t|^\downarrow_1 - |t|^\downarrow_2 &\leq \sigma^\downarrow_1(G_2) - \sigma^\downarrow_2(G_2).
\end{align}
In the second experiment, we consider the $3\times 3$ transmission matrix $G_3$ in Eq.~(\ref{eq:G_matrix_3_3_example}). We generate $\num{10000000}$ random $(U_{i},V_{i}) \in U(3)\times U(3)$ and calculate $\bm{t}[U_{i}, V_i]$ by Eq.~(\ref{eq:def_t_U_V}). Fig.~\ref{fig:t_abs}(a) shows the result of $|\bm{t}|[U_{i}, V_i]$. We see that $\{|\bm{t}|\}$ is a convex decahedron, which is defined by the set of solutions to a system of linear inequalities~\footnote{Inequality (\ref{ineq:redundant_3_3}) is redundant but included for later generalization.}:
\begin{align} 
|t|^\downarrow_1 &\leq \sigma^\downarrow_1(G_3), \\ \label{ineq:redundant_3_3}
|t|^\downarrow_1 + |t|^\downarrow_2 &\leq \sigma^\downarrow_1(G_3) + \sigma^\downarrow_2(G_3),   \\
|t|^\downarrow_1 + |t|^\downarrow_2 + |t|^\downarrow_3 &\leq \sigma^\downarrow_1(G_3) + \sigma^\downarrow_2(G_3) + \sigma^\downarrow_3(G_3),   \\
|t|^\downarrow_1 + |t|^\downarrow_2 - |t|^\downarrow_3 &\leq \sigma^\downarrow_1(G_3) + \sigma^\downarrow_2(G_3) - \sigma^\downarrow_3(G_3).
\end{align}

The numerical results above suggest the following observation on the geometry of $\{|\bm{t}|\}$: For an $n\times n$ transmission matrix $G$, $\{|\bm{t}|\}$ is a convex polytope in $\mathbb{R}^n$,  bounded by an intersection of half-planes. The half-planes are defined by a set of linear inequalities involving the coordinates of $\bm{\sigma}(G)$. We show that this observation is true as a result of our second theorem:

\begin{theorem}
Given a transmission matrix $G \in M_n$,  
\begin{align}
    \{\bm{t}\} = &\{\bm{v}\in \mathbb{C}^n \mid   | \bm{v}| \prec_w \bm{\sigma}(G), \notag \\  &\sum_{i=1}^{n-1} |\bm{v}|^\downarrow_i - |\bm{v}|^\downarrow_n \leq  \sum_{i=1}^{n-1} \bm{\sigma}^\downarrow_i(G) - \bm{\sigma}^\downarrow_n(G) \}.   \label{eq:main_result_t} 
\end{align}
\end{theorem}
\begin{proof}
This can be proved using the Sing-Thompson theorem~\cite{sing1976,thompson1977,guo2022a}. 
\end{proof}
Eq.~(\ref{eq:main_result_t}) is our third main result. It shows that $\{\bm{t}\}$ is completely determined by $\bm{\sigma}(G)$. Two media exhibit the same $\{\bm{t}\}$ if and only if they belong to the same $\bm{\sigma}$ class. We also note that only the magnitudes of $t_i$'s are constrained by Eq.~(\ref{eq:main_result_t}), while the phases of $t_i$'s can be arbitrarily set, e.g., by considering a diagonal unitary transformation $U$ or $V$.

\subsubsection{Answer to Question 4}

Finally, we discuss Question 4. The problem corresponds to the following physical scenario.  Suppose we have a medium characterized by a transmission matrix $G$. Given a direct transmission vector $\bm{t}_0 \in \{\bm{t}\}$, how do we construct the set of all possible unitary control schemes as described by unitary matrix pairs $\{(U,V)[\bm{t}_0]\}$ that achieve $\bm{t}_0$? Alternatively, a simpler question, how to construct one unitary control scheme as described by a unitary matrix pair $(U,V)[\bm{t}_0]$ that achieves  $\bm{t}_0$?

We don't have the answer to the first problem yet. We provide an algorithm that solves the second problem:

\begin{algorithm}[Constructing a ${(U,V)[\bm{t}_{0}]}$]
\hfill
\begin{enumerate}
\item Calculate the SVD of $G$: $G =  V_0^\dagger \Sigma U_0$. 
\item Use Chu's second algorithm~\cite{chu1999} (see Appendix~\ref{appendix:Chu_Fickus} for details) to construct a complex matrix $G'$ with singular values $\bm{\sigma}(G)$ and diagonal elements $\bm{t}_0$. Calculate the SVD of $G'$: $G' = V^{'\dagger} \Sigma U'$.
\item We obtain a $(U,V) [\bm{t}_0] = (U_0^\dagger U',V_0^\dagger V^{'})$.
\end{enumerate}\label{alg:a_UV}
\end{algorithm}
Algorithm~\ref{alg:a_UV} is our fourth main result. We illustrate its usage with a numerical example in Appendix~\ref{Appendix:algorithm3_demo}. We consider a $5\times 5$ transmission matrix $G$. Our task is to construct a $(U,V)[\bm{T}_0]$ with a randomly assigned goal $\bm{t}_0$. We use Algorithm~\ref{alg:a_UV} and complete the task. Importantly, our algorithm allows us to achieve the prescribed direct transmission in \emph{all} ports with a \emph{single} unitary matrix pair that performs the unitary control.

\section{Applications}\label{sec:applications}

Now we discuss the physical applications of our theory.

\subsection{Multimode coherent perfect/zero transmission}

First, we provide the criterion for $k$-fold degenerate coherent perfect transmittance, i.e.,~the effect that a medium exhibits perfect total transmittance for $k$ independent coherent input waves. From Eq.~(\ref{eq:main_result_T}), we obtain a necessary and sufficient condition for $G \in M_n$:
\begin{equation}
\bm{\sigma}^{\downarrow}_{1}(G)=\ldots =\bm{\sigma}^{\downarrow}_{k}(G)=1.    
\end{equation}
Similarly, we provide the criterion for $k$-fold degenerate coherent zero transmittance, i.e.,~the effect that a medium exhibits zero total transmittance for $k$ independent coherent input waves. From Eq.~(\ref{eq:main_result_T}), we obtain a necessary and sufficient condition for $G \in M_n$:
\begin{equation}
\bm{\sigma}^{\downarrow}_{1}(G)=\ldots =\bm{\sigma}^{\downarrow}_{k}(G)=0.    
\end{equation}

\subsection{Unitary uniform transmission}

Second, we propose the concepts of \emph{unitary uniform total transmittance}, and \emph{unitary uniform direct transmission}, which refer to the effects that a medium exhibits uniform total transmittance ($T_i = \text{const}$), and uniform direct transmission ($t_i = \text{const}$), respectively, under some unitary control. We claim that \emph{any} medium exhibits unitary uniform total transmittance. We prove this by showing that for any $G\in M_n$, there exist $U \in U(n)$ such that 
\begin{equation}
\bm{T}[U] =  (a, \ldots, a)^T, \quad a = \frac{1}{n}\sum_{i=1}^n \sigma_i^2(G).
\end{equation}
This is because for any $(x_1, x_2, \ldots,x_n)^T\in \mathbb{R}^n$, we have
\begin{equation}
(\bar{x},\bar{x},\ldots, \bar{x})^T \prec (x_1, x_2, \ldots,x_n)^T, \;\; \bar{x} = \frac{1}{n}\sum_{i=1}^n x_i.
\end{equation}
Thus, for any $G\in M_n$, we have
\begin{equation}
(a, \ldots, a)^T \prec \bm{\sigma}^2(G),  
\end{equation}
hence by Eq.~(\ref{eq:main_result_T}), it is attainable under unitary control. Moreover, we claim that \emph{any} medium exhibits unitary uniform direct transmission. We prove this by showing that for any $G\in M_n$, and for any $b\in \mathbb{C}$ that satisfies:
\begin{equation}
    |b| \leq \frac{1}{n}\sum_{i=1}^n \sigma_i(G), 
\end{equation}
there exist $(U,V) \in U(n)\times U(n)$ such that 
\begin{equation}
\bm{t}[U,V] =  (b, \ldots, b)^T.
\end{equation}
This is because $(b, \ldots, b)^T$ satisfies all the inequalities in Eq.~(\ref{eq:main_result_t}), hence it is attainable under unitary control. So, any transmission medium can exhibit uniform total transmittance or direct transmission over any number of ports under suitable unitary control. Moreover, our algorithms provide practical implementations to achieve these effects. These results can be useful in applications such as uniform illumination.

\subsection{Symmetry constraints on bilateral transmission}

Third, we discuss the constraints imposed by symmetry on the unitary control of bilateral transmission. So far, our focus has been on the transmission matrix $G$ in the forward direction (from left to right). However, in many applications, it is also necessary to consider the transmission matrix $\tilde{G}$ in the backward direction (from right to left). In general, $G$ and $\tilde{G}$ can be distinct. Nevertheless, certain symmetries of the system can establish a relationship between them, thereby affecting their unitary control. Here we examine two significant internal symmetries~\cite{zhao2019c,guo2022c}: reciprocity and energy conservation. We have proven that when the medium is either reciprocal or energy-conserving, 
\begin{equation}\label{eq:relation_G_G_tilde}
\bm{\sigma}(\tilde{G}) = \bm{\sigma}(G).
\end{equation}
(See Appendix~\ref{appendix:proof} for proof.) Consequently, $G$ and $\tilde{G}$ belong to the same $\bm{\sigma}$ class and exhibit the same sets of attainable $\{\bm{T}\}$ and $\{\bm{t}\}$ under unitary control.

\section{Conclusion}\label{sec:conclusion}

In conclusion, we provide a systematic theory for unitary control of wave transmission. We reveal that singular value inequalities provide the mathematical structure to describe the physics of unitary control.  Our results deepen the understanding of unitary control of wave transmission and provide practical guidelines for its implementation.

\begin{acknowledgments}
This work is funded by a Simons Investigator in Physics
grant from the Simons Foundation (Grant No.~827065) and
by a Multidisciplinary University Research Initiative
(MURI) grant from the U.S. Air Force Office of Scientific Research (AFOSR) (Grant No.~FA9550-21-1-0312).

\end{acknowledgments}
%\nocite{*}

%\clearpage
%\onecolumngrid
\appendix

\section{Chu's algorithms and Fickus' algorithm}\label{appendix:Chu_Fickus}

Here we briefly review Chu's and Fickus' algorithms for constructing a Hermitian matrix with prescribed diagonal entries and eigenvalues.

\begin{algorithmappendix}[Chu, 1995~\cite{chu1995}]\label{alg:Chu}
One can construct a real symmetric matrix with prescribed eigenvalues $\bm{\lambda}$ and diagonal entries $\bm{d}$ by integrating the differential equations:
\begin{equation}
\dot{X} = \left[X,\left[\operatorname{diag}(X) - \operatorname{diag}(\bm{d}), X \right]\right]    
\end{equation}
till equilibrium from a starting point $X_{0} = Q^{T}\Lambda Q$ with $Q$ a random orthogonal matrix and $\Lambda = \operatorname{diag}(\bm{\lambda})$. Here $[A, B] \equiv AB-BA$ is the Lie bracket, $\operatorname{diag}(X)$ is the diagonal matrix with the same diagonal entries of $X$, and $\operatorname{diag}(\bm{d})$ is the diagonal matrix with diagonal entries $\bm{d}$. This algorithm always converges to a valid solution.
\end{algorithmappendix}
\begin{algorithmappendix}[Fickus, 2013~\cite{fickus2013}]\label{alg:Fickus}
One can construct all Hermitian matrices with prescribed eigenvalues $\bm{\lambda}$ and diagonal entries $\bm{d}$ using finite frame theory. The explicit steps can be found in Ref.~\cite{fickus2013}.
\end{algorithmappendix}
We also review Chu's algorithm for constructing a matrix with prescribed diagonal entries and singular values. 
\begin{algorithmappendix}[Chu, 1999~\cite{chu1999}]\label{alg:Chu_2nd}
One can construct a real matrix with prescribed singular values $\bm{\sigma}$ and diagonal entries $\bm{d}$ using a recursive method. The explicit codes can be found in Ref.~\cite{chu1999}.
\end{algorithmappendix}

\section{Demonstration of Algorithm~\ref{alg:a_U}}\label{Appendix:algorithm2_demo}

To illustrate Algorithm~\ref{alg:a_U}, we consider a $5\times 5$ transmission matrix 
\begin{widetext}
\begin{equation}\label{eq:G_example}
G =  \begin{pNiceMatrix}
-0.15-0.13i & -0.09 -0.02i & -0.13 - 0.35i & -0.13-0.01i & -0.20 - 0.12i \\
0.36 -0.20i & 0.14 + 0.10i & -0.08-0.10i & -0.16 +0.24i & 0.22 -0.34 i \\
0.44 - 0.14 i & -0.11 - 0.05i & -0.07 +0.22i & -0.04 + 0.15i & -0.19 -0.19i \\
0.02 +0.00i & -0.04 + 0.19i & 0.25 + 0.12i & 0.11 - 0.02 i & 0.18 - 0.12i \\
0.04 -0.17 i & -0.45 -0.13 i & -0.25 - 0.14 i & 0.07 - 0.03 i & -0.18 - 0.07i
\end{pNiceMatrix},
\end{equation}
\end{widetext}
with
\begin{align}
 \bm{\sigma}(G) &= \begin{pNiceMatrix}
 0.90, &0.70, &0.50, &0.30, &0.10
 \end{pNiceMatrix}^{T}; \\
 \bm{\sigma}^{2}(G) &= \begin{pNiceMatrix}
 0.81, &0.49, &0.25, &0.09, &0.01
 \end{pNiceMatrix}^{T}.
\end{align}
The task is to construct a ${U[\bm{T}_{0}]}$, with the randomly assigned goal
\begin{align}
\bm{T}_{0} &= \begin{pNiceMatrix}
0.55, &0.20, &0.28, &0.15, &0.46
\end{pNiceMatrix}^{T}.
\end{align}
First, we check that
\begin{equation}
\bm{T}_{0} \prec \bm{\sigma}^{2}(G),
\end{equation}
so $\bm{T}_0$ is attainable via unitary control. We use Algorithm~\ref{alg:a_U} to obtain:
\begin{widetext}
\begin{equation}
U[\bm{T}_{0}] = \begin{pNiceMatrix}
-0.23 + 0.00i & 0.53 +0.00i & 0.23 + 0.00i &-0.31+0.00i & 0.72+0.00i \\
0.10-0.48i & -0.37+0.01i & -0.38 - 0.09 i & -0.62 + 0.26i & -0.15+0.02i \\
-0.48 - 0.44 i & 0.08 -0.04 i & -0.35 -0.34i & 0.30 -0.44i & -0.03 + 0.19 i \\
-0.09 - 0.17i & 0.35 + 0.32i & 0.36 - 0.25i & -0.40 + 0.00 i & 0.58 + 0.21 i \\
-0.32 -0.37i & -0.22 - 0.54i & 0.52 + 0.32i & -0.05 -0.05i & 0.12 - 0.17i
\end{pNiceMatrix}. \notag 
\end{equation}
\end{widetext}
We verify that 
\begin{equation}
\bm{T}[U] = \bm{d}(U^\dagger[\bm{T}_{0}] G^\dagger G U[\bm{T}_{0}]) =  \bm{T}_{0}.  
\end{equation}

\section{Demonstration of Algorithm~\ref{alg:a_UV}}\label{Appendix:algorithm3_demo}

To illustrate Algorithm~\ref{alg:a_UV}, we consider the same $G\in M_5$ as given in Eq.~(\ref{eq:G_example}). The task is to construct a $(U,V)[\bm{t}_{0}]$, with the randomly assigned goal
\begin{align}
\bm{t}_{0} &= \begin{pNiceMatrix}
0.50, &-0.20, &0.10, &0.30, &-0.40
\end{pNiceMatrix}^{T}.
\end{align}
First, we check that
\begin{equation}
|\bm{t}_{0}| \prec_w \bm{\sigma}(G), \qquad \sum_{i=1}^{n-1} |\bm{t}_{0}|^\downarrow_i - |\bm{t}_{0}|^\downarrow_n \leq  \sum_{i=1}^{n-1} \bm{\sigma}^\downarrow_i(G) - \bm{\sigma}^\downarrow_n(G).
\end{equation}
so $\bm{t}_0$ is attainable via unitary control. We use Algorithm~\ref{alg:a_UV} to obtain:
\begin{widetext}
\begin{equation}
U[\bm{t}_{0}] = \begin{pNiceMatrix}
0.27 + 0.00i & -0.41 +0.00i & 0.37 + 0.00i &-0.52+0.00i & -0.60+0.00i \\
-0.67 + 0.38i & -0.00+0.22i & -0.23 - 0.29 i & -0.43 - 0.03i & -0.07+0.17i \\
0.41 - 0.03i & -0.14 +0.08 i & -0.51 -0.34i & -0.06 -0.60i & 0.02 + 0.24 i \\
0.20 + 0.30i & 0.43 + 0.28i & 0.48 - 0.19i & -0.37 - 0.19 i & 0.41 - 0.01i \\
0.07 -0.17i & -0.33 - 0.62i & 0.30 + 0.02i & -0.05 -0.03i & 0.48 + 0.39i
\end{pNiceMatrix}. \notag 
\end{equation}
\begin{equation}
V[\bm{t}_{0}] = \begin{pNiceMatrix}
-0.19 - 0.43i & 0.25 -0.38i & -0.60 + 0.35i & 0.14+0.21i & -0.05+0.16i \\
-0.43-0.34i & -0.02-0.14i & 0.19 - 0.08 i & -0.70 + 0.07i & -0.05 - 0.37i \\
0.20 + 0.17 i & -0.13 +0.07 i & -0.15 + 0.15i & -0.39 + 0.39i & 0.72 + 0.22 i \\
0.18 - 0.23i & -0.24 + 0.61i & -0.03 + 0.51i & -0.21 + 0.10 i & -0.40 + 0.11 i \\
0.49 -0.28i & -0.26 - 0.51i & 0.41 - 0.06i & -0.02 +0.29i & -0.20 + 0.24i
\end{pNiceMatrix}. \notag 
\end{equation}
\end{widetext}
We verify that 
\begin{equation}
\bm{t}[U,V] = \bm{d}(V^\dagger[\bm{t}_{0}] G U[\bm{t}_{0}]) =  \bm{t}_{0}.  
\end{equation}

\section{Proof of Eq.~(\ref{eq:relation_G_G_tilde})}\label{appendix:proof}

To prove Eq.~(\ref{eq:relation_G_G_tilde}), we consider the whole scattering matrix of the medium:
\begin{equation}
S = \begin{pmatrix}
R_{1} & \tilde{G}  \\
G & R_{2}
\end{pmatrix} \in M_{2n}    
\end{equation}
where $R_{1}$ and $R_{2}$ are the reflection matrices on the left and right sides, respectively. 

If the system is reciprocal, then~\cite{guo2022c} 
\begin{equation}
S = S^T,    
\end{equation}
which implies that
\begin{equation}
\tilde{G} = G^T,    
\end{equation}
thus 
\begin{equation}
\bm{\sigma}(\tilde{G}) = \bm{\sigma}(G).    
\end{equation}

If the system is energy-conserving, then~\cite{guo2022c}
\begin{equation}
S^{\dagger}S = S S^{\dagger} = I_{2n}.    
\end{equation}
From $S^{\dagger}S =  I_{2n}$, we obtain
\begin{equation}
R_{1}^{\dagger}R_{1} + G^{\dagger}G= I_{n}.    
\end{equation}
From $SS^{\dagger}=  I_{2n}$, we obtain
\begin{equation}
R_{1} R_{1}^{\dagger} + \tilde{G} \tilde{G}^{\dagger}= I_{n}. 
\end{equation}
Therefore,
\begin{align}
\bm{\sigma}^2(G) = \bm{\lambda}(G^{\dagger}G) = \bm{\lambda}(I_n - R_{1}^{\dagger}R_{1}) = \bm{1}- \bm{\sigma}^2(R_{1}) ; \\ 
\bm{\sigma}^2(\tilde{G}) = \bm{\lambda}(\tilde{G}\tilde{G}^{\dagger}) = \bm{\lambda}(I_n - R_{1}R_{1}^{\dagger}) = \bm{1}- \bm{\sigma}^2(R_{1}).
\end{align}
Thus,
\begin{equation}
\bm{\sigma}(\tilde{G}) = \bm{\sigma}(G).    
\end{equation}

This completes the proof of Eq.~(\ref{eq:relation_G_G_tilde}).

\bibliography{main}% Produces the bibliography via BibTeX.

\end{document}